\documentclass[10pt,USenglish]{article}

\usepackage[a4paper, total={6in, 8in}]{geometry}

\usepackage{graphicx}
\usepackage{xspace}
\usepackage{comment}
\usepackage{amsthm,amsmath}
\usepackage[hidelinks]{hyperref}
\theoremstyle{plain}
\usepackage{authblk}

\newtheorem*{lemmawithoutnumber}{Lemma}
\newtheorem*{theoremwithoutnumber}{Theorem}
\newtheorem{definition}{Definition}
\newtheorem{lemma}{Lemma}
\newtheorem{theorem}{Theorem}
\newtheorem{corollary}{Corollary}

{\itshape}{\rmfamily}
\newcommand{\prob}[1]{\textsc{#1}}

\newcommand{\MACD}{\prob{Maximum Alternating Cycle Decomposition}\xspace}
\newcommand{\macd}{\prob{MACD}\xspace}

\newcommand{\SSBR}{\prob{Signed Sorting By Reversals}\xspace}
\newcommand{\ssbr}{\prob{SSBR}\xspace}

\newcommand{\SBR}{\prob{Sorting By Reversals}\xspace}
\newcommand{\sbr}{\prob{SBR}\xspace}

\newcommand{\MLS}{\prob{Minimum Cost Scenario}\xspace}

\newcommand{\MLPS}{\prob{Minimum Cost Parsimonious Scenario}\xspace}
\newcommand{\mlps}{\hyperref[probMLPS]{\prob{MCPS}}\xspace}

\newcommand{\MISA}{\prob{Maximum Independent Set of Arcs}\xspace}
\newcommand{\misa}{\prob{MISA}\xspace}

\newcommand{\eg}{\emph{e.g.}\xspace}

\title{A framework for cost-constrained genome rearrangement under Double Cut and Join}

\author[1]{Pijus Simonaitis}
\author[1,2]{Annie Chateau}
\author[1,2]{Krister M. Swenson}
\affil[1]{LIRMM, CNRS -- Universit\'e Montpellier\\
  161 rue Ada, 34392 Montpellier, France\\}
\affil[2]{Institut de Biologie Computationnelle (IBC)\\
  Montpellier, France}
\date{}       

\begin{document} 

\maketitle

\begin{abstract} 

The study of genome rearrangement has many flavours, but they all are somehow
tied to edit distances on variations of a multi-graph called the breakpoint
graph.
We study a weighted 2-break distance on Eulerian 2-edge-colored multi-graphs,
which generalizes weighted versions of several Double Cut and Join problems,
including those on genomes with unequal gene content.
We affirm the connection between cycle decompositions and edit scenarios
first discovered with the \SBR problem.  Using this we show that 
the problem of finding a parsimonious scenario of minimum cost on an Eulerian
2-edge-colored multi-graph -- with a general cost function for 2-breaks --  can
be solved by decomposing the problem into independent instances on simple
alternating cycles.
For breakpoint graphs, and a more constrained cost function, based on coloring
the vertices, we give
a polynomial-time algorithm for finding a parsimonious 2-break scenario of
minimum cost, while showing that finding a non-parsimonious 2-break scenario of
minimum cost is NP-Hard.

\end{abstract}

\section{Introduction}

Edit distance problems are a mainstay in computer science.  On strings, edit
distances have roots in computational linguistics, and are at the heart of many
approximate string matching algorithms for signal processing and text
retrieval~\cite{navarro2001guided}.
The Levenshtein distance is the classic example, which asks for a minimum number of
insertions, deletions, or substitutions of characters needed to transform one
string into another.
On graphs, the first edit distance to be considered is an analogue to the
Levenshtein distance; insertion and deletion of vertices and edges are allowed,
along with vertex and edge label
substitution~\cite{sanfeliu1983distance}.
Graph edit distances have become an important tool in pattern recognition, which
is at the heart of modern image processing and computer vision~\cite{bunkegraph,gao2010survey}.

From early on, there were major applications for string edit distances in
computational biology, mainly due to the linear nature of DNA, RNA, and protein
molecules
~\cite{needleman1970general,sankoff1972matching,kruskal1983time,gusfield1997algorithms}.
More recently, graph edit distances have found a role in the comparison of gene
regulatory networks~\cite{mcgrane2016biological}.
In the next subsection we outline the pervasive role of edit distances in genome
rearrangements.

In this paper we address a graph edit distance on Eulerian 2-edge-colored
multi-graphs, that is, a multi-graph with black and gray edges
such that every vertex is incident to the same number of black and gray edges.
When we say \emph{graph}, we will mean Eulerian 2-edge-colored multi-graph.
The edit operations that we consider are the \emph{k-break}, where $k$ black
edges are replaced such that the degree of all vertices in the graph are
conserved.
For 2-breaks, we generalize the edit distance problem to consider costs on
the operations, posed as the \MLPS (\mlps) problem.
This problem asks for a minimum-length scenario of 2-breaks transforming the
set of black edges into the set of gray edges.

Our main contribution is a clean formalism that facilitates simple proofs
for edit distances on 2-breaks with costs.
While weighted edit operations have been considered in the past, to our
knowledge, this is the first study of weighted 2-breaks in this general
setting~\cite{benderGHHPSS2008, Farnoud(Hassanzadeh):2012:SPC:2331863.2332291, Swenson2016,simonaitis_et_al:LIPIcs:2017:7660}.

In Section~\ref{decompositionOFscenario} we show that a k-break scenario
on a graph $G$ partitions $G$ into a set of Eulerian subgraphs, such that no
k-break operates on edges from different subgraphs.
This decomposition theorem allows us to show a strong link between \MACD
and the length of a parsimonious 2-break scenario for a graph.
It is also the cornerstone for Section~\ref{generalmlps}, showing that
\mlps can be computed on a graph $G$ if there is a method for computing \mlps
on a circle.
By \emph{circle} we mean a graph where all vertices are incident to exactly
one gray and one black edge (Section~\ref{generalmlps}).

These results are general in the sense that the variants of the breakpoint
graph typically used for genome rearrangement problems are all 
specific instance of the Eulerian 2-edge-colored multi-graph.
Section~\ref{secEditDistancesAndRearrangements} gives such examples.

Section~\ref{colorcost} is dedicated to a specific cost function that
depends on a coloring of the vertices:  if a 2-break replaces edges $(a,b)$ and
$(c,d)$ with $(a,d)$ and $(c,b)$, then the cost is zero if $a$ and $c$ have the
same color, or $b$ and $d$ have the same color, otherwise it is of cost 1.
Our tool for reasoning with this cost function is called the \emph{color-merged graph}, which is obtained from a graph by merging all vertices of the same
color.
Theorem~\ref{mcsnphard} states that it is NP-Hard to compute \MLS, even for
circles.
After establishing the necessary links between cycle decompositions in
a graph and its merged graph (Section~\ref{junctiondecomposition}), we show
how to compute \mlps on a circle (Sections~\ref{circlecost} and
\ref{secPolyMCPS}).
Finally, we show that \MLPS for the colored cost function is computable in
$O(n^4)$ time.

\subsection{Edit distances and breakpoint graphs}
\label{secEditDistancesAndRearrangements}
In the area of gene order comparison through gene rearrangements, both the
string and graph edit distances play a central
role~\cite{fertin2009combinatorics}.
The typical genome rearrangement problem is an edit distance problem
on strings of genes called \emph{genomes}, where each gene occurs exactly once.
The first biologically motivated edit operation that was studied is the
\emph{reversal} of a substring of a genome~\cite{watterson1982,Sankoff92}.
When the relative direction of gene transcription is known, this problem is
called \SSBR (\ssbr), and is intimately linked to the \emph{breakpoint
graph}~\cite{kececioglu1993exact,bafna1993genome}.
Roughly speaking, this multi-graph has a vertex for each gene extremity, and for
every pair of adjacent gene extremities there is an edge.
Thus, the edges for a single genome constitute a perfect matching.
Say there is a gray genome and a black genome, then we have a gray matching
and a black matching in the breakpoint graph.
In this case the graph is 2-regular, and therefore decomposes into disjoint
cycles of alternating color.
A reversal on the black genome replaces two adjacencies.
In this way, \ssbr can be seen as a graph edit problem where
the graph edit operation is a replacement of 2 edges, simulating the reversal
of a substring.
The \ssbr problem can be solved in
polynomial-time~\cite{hannenhalli1999transforming}.

When relative transcription directions are unknown for the genes, we have the
\SBR (\sbr) problem.
In \sbr, since the orientation of gene extremities are unknown, the vertices
for the two extremities of each gene are merged into a single vertex in the
breakpoint graph.
The result is a 4-regular graph.
While finding the alternating-color cycle decomposition of the 2-regular
breakpoint graph is trivial, finding the same for a 4-regular graph is not:
Caprara showed that \sbr is as hard as finding a \MACD (\macd), and that 
\macd is NP-Hard~\cite{caprara1997sorting}.


Currently the most mathematically clean model for genome rearrangement is
called the \emph{Double Cut and Join} (DCJ)
model~\cite{DCJ,Bergeron2006}.
Genome extremities that are adjacent are paired, and transformations of these
pairs occur by swapping elements of the pairs.
DCJ is a generalization of the \ssbr paradigm, since a DCJ operation on a
genome can simulate a reversal.
This implies that DCJ is a less restrictive graph edit distance model, where
edge pair $\{(a,b), (c,d)\}$ is replaced by either $\{(a,c), (b,d)\}$ or
$\{(a,d), (b,c)\}$ (in \ssbr only one of the two replacements would be a valid
reversal).
The DCJ edit distance is inversely proportional to the number of cycles in the
breakpoint graph.

While much of the work on genome rearrangement is on models where exactly one
occurrence of each gene exists in each genome, content modifying operations
like gene insertion, deletion, and duplication have been considered.
Approaches to these problems are inextricably tied to finding cycle decompositions on breakpoint graphs~\cite{el2000genome,braga2011double,shao2015exact,shao2015comparing}.
In the 2-break model, the \MACD for each of these graphs implies the 2-break
distance.
Figure~\ref{figBPGraph} shows breakpoint graphs (see
Section~\ref{mlpsbreakpoint} for a definition) for genomes with unequal
gene content.
Our results on 2-breaks, up to and including Section~\ref{secPolyMCPS}, are
directly applicable to the DCJ model, even in the presence of content modifying
operations.

\begin{figure}[]
\centering
\includegraphics[scale = 1]{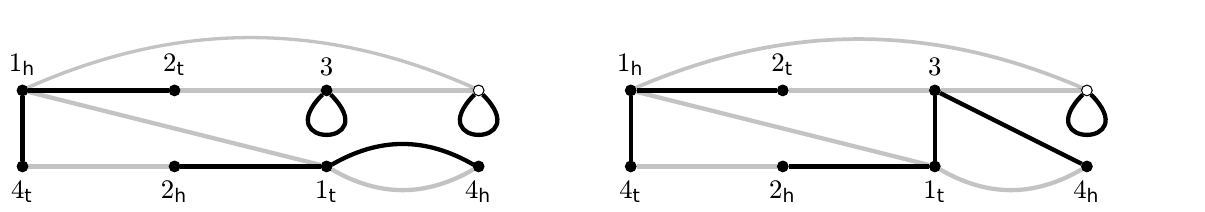} 
\caption{Examples of breakpoint graphs; $G(A,B)$ on the left and $G(A',B)$ on the right, where
$A=\{\{1_{h},2_{t}\},\{2_{h},1_{t}\},\{1_{h},4_{t}\},\{4_{h},1_{t}\}\}$,
$A'=\{\{1_{h},2_{t}\},\{2_{h},1_{t}\},\{1_{h},4_{t}\},\{4_{h},3\},\{3,1_{t}\}\}$, and
$B=\{\{3\},\{3,2_{t}\},\{2_{h},4_{t}\},\{4_{h},1_{t}\},\{1_{h}\}\}$.
Edges adjacent to a special vertex $\circ$ represent the endpoints of linear
chromosomes (\eg gray edges $\{3,\circ\}$ and $\{1_h,\circ\}$), while missing
genes are represented by self-loops (\eg the black edge $(3,3)$ and the gray edge
$(1_h,1_t)$), called \emph{ghost adjacencies} in~\cite{Shao2012}.
In the circular genomes $A$ and $A'$, gene $1$ is repeated twice, and the DCJ transforming $A$ into $A'$ is an insertion of an unoriented gene $3$, corresponding to the 2-break $G(A,B)\rightarrow G(A',B)$. A DCJ scenario transforming $A'$ into the linear genome $B$ includes a deletion of one copy of a gene 1.}
\label{figBPGraph}
\end{figure}

\subsection{Distinguishing features of our model}
\label{secDistFeatures}
We have previously worked on DCJ problems with weight functions~\cite{Swenson2016,simonaitis_et_al:LIPIcs:2017:7660}.
These papers address the 2-break model for 2-regular breakpoint graphs with a
cost function on the edges instead of the vertices.
The differences between the vertex-cost model from this paper, and edge-cost
model of previous work are subtle.
Indeed, the algorithmic results from Sections~\ref{circlecost},\ref{secPolyMCPS}, and~\ref{mlpsbreakpoint} are
similar to those in~\cite{Swenson2016}, but bear the hallmark of this vertex
colored model: simplicity.
This simplicity facilitates the much more general nature of this work.

From a practical perspective, the vertex-cost model we present here sidesteps
the weakest feature of the edge-cost model: the cost function.
In the edge-cost model, each DCJ on a particular pair of adjacencies
has two possible states.
Which state is chosen determines the cost of subsequent DCJs on these
adjacencies.
In this way, the edge-cost function is state dependent.

For our vertex-cost model, the natural bijection between vertices representing
the same extremity in the two genomes allows for a stateless cost function.
This way, the weight of a DCJ can be computed as function of both genomes
instead of just one.
We also speculate that the computation of a general cost function under the
vertex-cost model will be simpler.
The only known polynomial time algorithm for the edge-cost model is restricted
to a color based cost function.

\section{k-breaks on a 2-edge-colored graph}

\subsection{Definitions}
In our work a \emph{graph} will be an Eulerian 2-edge-colored undirected multi-graph $G=(V,E^{b}\cup E^{g})$ with black and gray edges. We set $e(G)=|E^{b}|=|E^{g}|$. 
 
\begin{definition}[Eulerian graphs and alternating cycles] 
For a graph $G$ and a vertex $v$ we set $\bar{d}(G,v)=d^{b}(G,v)-d^{g}(G,v)$ where $d^{b}$ and $d^{g}$ are black and gray degrees of $v$.
$G$ is \emph{Eulerian} if $\bar{d}(G,v)=0$ for every vertex $v$. 
A cycle is \emph{alternating} if it is Eulerian. 
All the cycles in our work will be alternating unless specified otherwise.
The \emph{length} of a cycle is its number of black edges.
\end{definition}

\begin{definition}[k-break] A \emph{k-break} is a transformation of a graph $G$ into $G'$ that replaces $k$ black edges $(x_{1},x_{2}),\ldots, (x_{2k-1}, x_{2k})$ by $(x_{q_{1}},x_{q_{2}}),\ldots, (x_{q_{2k-1}}, x_{q_{2k}})$ while preserving the degree of all the vertices of $G$.  In other words, $d^{b}(G,v)=d^{b}(G',v)$ for all $v$ and the multi-sets $\{x_{q_{1}}, \ldots, x_{q_{2k}}\}$ and $\{x_{1}, \ldots, x_{2k}\}$ are equal.
\end{definition}

Since $G$ is Eulerian, it admits a decomposition into edge-disjoint alternating cycles. 
We denote the size of a \MACD (\macd) by $c(G)$ .   
If $c(G)=1$ we say that $G$ is a \emph{simple cycle}. All the cycles in a \macd are simple.
A graph with equal multi-sets of black and gray edges will be called \emph{terminal}. 
A \emph{scenario} will be a sequence of k-breaks transforming a graph into a terminal graph.

\subsection{Cycle decompositions for a k-break scenario}
\label{decompositionOFscenario}

To a k-break scenario we will associate a cycle decomposition such that all the edges replaced by a k-break belong to a single cycle.
In Section~\ref{generalmlps} this will allow us to concentrate on the scenarios for the simple cycles instead of the general graphs.

We say that two edges are \emph{$l$-related} according to a given scenario if among the first $l$ k-breaks of the scenario there is one replacing these two edges.
$P_{l}$ will be a set of the equivalence classes of the transitive closure of this relation. 
We start by constructing $P_{l}$, which is a partition of the black edges of a graph, and proceed by showing how the gray edges can be incorporated to obtain a cycle decomposition of a graph.   

Take a k-break scenario of length $m$ for a graph $G=G_{0}$. 
We denote the graph after $l\geq 0$ k-breaks by $G_{l}$.
We label the $G_{0}$ black edges by the set $\{1,\ldots, e(G)\}$ and partition them into singletons $P_{0}=\{\{1\},\ldots,\{e(G)\}\}$. 
The $l$th k-break replaces $k$ black edges labeled $e_{1}, \ldots, e_{k}$ of $G_{l-1}$ by $k$ black edges 
that we arbitrarily label $e_{1},\ldots,e_{k}$.
Let $Q_{i}$ be the subset of a partition $P_{l-1}$ including the edge $e_{i}$.
We obtain a new partition $P_{l}$ from $P_{l-1}$ by merging all $Q_i$ into the set $\bigcup_{i=1}^{k} Q_{i}$. 
After applying all $m$ merges of a scenario we obtain a partition $P_{m}$ of the set $\{1,\ldots, e(G)\}$.

We label the gray edges of $G_{m}$ by $\{e(G)+1, \ldots, 2e(G)\}$, keeping the same labeling for all $G_{l}$.
Since $G_{m}$ is terminal, its multi-sets of black and gray edges are equal.
This implies a one-to-one mapping of labels between $\{1, \ldots, e(G)\}$ and $\{e(G)+1, \ldots, 2e(G)\}$ such that corresponding labels have the same endpoints in $G_{m}$. 
Using this mapping, we produce a partition on $\{1,\ldots, 2e(G)\}$ called $R_l$, that includes all black edge labels $e_i \in P_l$ along with the gray edge label that $e_i$ maps to.
We will show that every subset of $R_{m}$ defines an Eulerian subgraph of $G$.
stack
For a subset $Q\subset\{1,\ldots, 2e(G)\}$ and $l\in\{0,\ldots,m\}$,
we define $s(G_{l}, Q)$ to be the edge-induced subgraph of $G_{l}$ having only the edges of $G_{l}$ labeled by elements of $Q$. 
By construction, for $Q\in R_{m}$ and a vertex $v$ we have $\bar{d}(s(G_{m},Q),v)=0$. 
Lemma~\ref{degrees}, proven in the Appendix, establishes the equality $\bar{d}(s(G_{0},Q),v)=\bar{d}(s(G_{m},Q),v)$, which means that $s(G,Q)$ is Eulerian.    
For a scenario $\rho$ $C(\rho) = R_{m}$ denotes the \emph{cycle decomposition} of the scenario $\rho$.

\begin{lemma}
\label{degrees}
$\bar{d}(s(G_{0},Q),v)=\bar{d}(s(G_{l},Q),v)$ for a vertex $v$ and $Q\in R_{l}$ with $l\in\{0,\dots m\}$.
\end{lemma}   

\begin{theorem}
\label{minlength2break}
The minimum length of a 2-break scenario $l(G)$ is equal to $e(G)-c(G)$.
\end{theorem}
\begin{proof}

We first show that for a 2-break scenario $\rho$ of length $m$ we have $|C(\rho)|\geq e(G)-m$.   
Recall that the size of $C(\rho)$ is equal to $|R_{m}|=|P_{m}|$ where $P_{m}$ is a partition of a set $\{1,\ldots, e(G)\}$ encountered in the construction of $C(\rho)$. 
The $l$th 2-break of $\rho$ replaces two edges $e_{1}$ and $e_{2}$ and merges the subsets $Q_{1}$ and $Q_{2}$ of $P_{l-1}$ (recall that $e_{1} \in Q_{1}$ and $e_{2}\in Q_{2}$) to obtain a partition $P_{l}$. 
By construction, $|P_{l}|\geq |P_{l-1}|-1$ as at most two subsets get merged. 
The size of $P_{0}$ is equal to $e(G)$, thus the size of $P_{m}$ is at least $e(G)-m$, meaning that $c(G)\geq e(G)-l(G)$.    

On the other hand, for any cycle $c$ of length $l > 1$
there is a 2-break transforming $c$ into a union of length 1 and length $l-1$ cycles. 
In this way we obtain a scenario of length $l-1$ for $c$,
and can transform every cycle of a \macd of $G$ independently, obtaining a 2-break scenario of length $e(G)-c(G)$. Thus, $l(G)\leq e(G)-c(G)$. 
\end{proof}

\begin{corollary} $C(\rho)$ for a parsimonious 2-break scenario $\rho$ is a \MACD of $G$.
\end{corollary}

\section{\MLPS}
\label{generalmlps}

Consider a non-negative cost function $\varphi$ for the 2-breaks on $V$.
By \emph{2-breaks on} $V$ we mean the set of all the 2-breaks on the complete graph with vertices $V$.
The cost of a scenario on a graph $G=(V,E^{b}\cup E^{g})$ is the sum of the costs of its 2-breaks.
We provide an example of a cost function $\varphi$ in Section~\ref{colorcost}.
The \MLPS (\mlps) for a graph under cost function $\varphi$ is a minimum cost scenario among the scenarios for $G$ of minimum length. 
$\mlps_{\varphi}(G)$ denotes the cost of a \mlps for a graph $G$ and a cost function $\varphi$.

\begin{theorem}
\label{sumofcosts}
$\mlps_{\varphi}(G)$ is the minimum over $\big\{\sum_{c \in C} \mlps_\varphi(c) ~\big|~ $C$~\text{is a \macd of}~$G$ \big\}$.
\end{theorem}
\begin{proof}
Take a \mlps $\rho$ for $G$ and $\varphi$. 
$C(\rho)$ is a \macd of $G$ due to Theorem~\ref{minlength2break}.
The subsequence of $\rho$ consisting of the 2-breaks replacing the edges of a cycle $c\in C(\rho)$ is a parsimonious scenario for $c$ that we name $\rho_{c}$.
A 2-break sequence $\rho'$ obtained by performing $\rho_{c}$ one by one is a parsimonious scenario for $G$.
The costs of $\rho$ and $\rho'$ are equal as they consist of the same 2-breaks performed in different order. 
This way we know that the cost of $\rho$ is smaller or equal to $\sum_{c\in C(\rho)} \mlps_{\varphi}(c)$.
On the other hand, for a \macd $C$ minimizing the sum $\sum_{c\in C} \mlps_{\varphi}(c)$ we can construct a parsimonious scenario of cost $\sum_{c\in C} \mlps_{\varphi}(c)$ transforming each cycle separately. 
\end{proof}

\subsection{\mlps for a simple cycle}
\label{simplecost}

A simple cycle $S$ might have a certain number $d(S)$ of vertices $v$ with $d^{b}(S,v)=d^{g}(S,v)=2$.  It is easy to check that $d^{b}(G,v)=d^{g}(G,v)<3$ for any vertex. 
If $d(S)=0$, then we call $S$ a \emph{circle}. 
See Figure~\ref{simplecycles} for an example of a simple cycle that is not a circle. 

Take a simple cycle $S$ with $d(S)>0$, a cost function $\varphi$ for the 2-breaks on its vertices and a vertex $v_{0}$.
For every Eulerian cycle $(v_{0},\ldots,v_{m-1},v_{0})$ of $S$ we construct a circle $(u_{0},\ldots,u_{m-1},u_{0})$ with a cost function $\varphi'$ defined as follows
for $i,j,k,l\in\{0,\ldots,m-1\}$:
\begin{align*}
\varphi'((u_{i},u_{j}),(u_{k},u_{l})\rightarrow(u_{i},u_{k}),(u_{j},u_{l}))=\varphi((v_{i},v_{j}),(v_{k},v_{l})\rightarrow(v_{i},v_{k}),(v_{j},v_{l})).
\end{align*}
There are no more than $2^{d(S)}$ of such circles. In Figure~\ref{simplecycles} a simple cycle $S$ is given together with its two Eulerian circles.
A simple cycle $S$ is recovered from its Eulerian circle by merging $d(S)/2$ pairs of vertices. 
In the Appendix we prove the following theorem:
\begin{theorem} 
\label{costOFsimple}
The $\mlps_{\varphi}$ cost of a simple cycle $S$ is equal to the minimum over all $\mlps_{\varphi'}$ costs of its Eulerian circles. 
\end{theorem}

Given a subroutine computing $\mlps_{\varphi}$ for a circle we can compute $\mlps_{\varphi}$ for every simple cycle of $G$ using Theorem~\ref{costOFsimple}.
Then $\mlps_{\varphi}(G)$ can be computed by choosing a \macd of $G$ maximizing the sum of the costs of its cycles using Theorem~\ref{sumofcosts}. 
In Section~\ref{colorcost} we define a particular cost function $\varphi$ and provide a polynomial time algorithm computing $\mlps_{\varphi}$ for a circle in Section~\ref{secPolyMCPS}. Then in Section~\ref{mlpsbreakpoint} we show how, using this algorithm as a subroutine, a minimum cost parsimonious DCJ scenario transforming genome $A$ into $B$ can be found in polynomial time. 

\begin{figure}[]
\centering
\includegraphics[width = 0.8\linewidth]{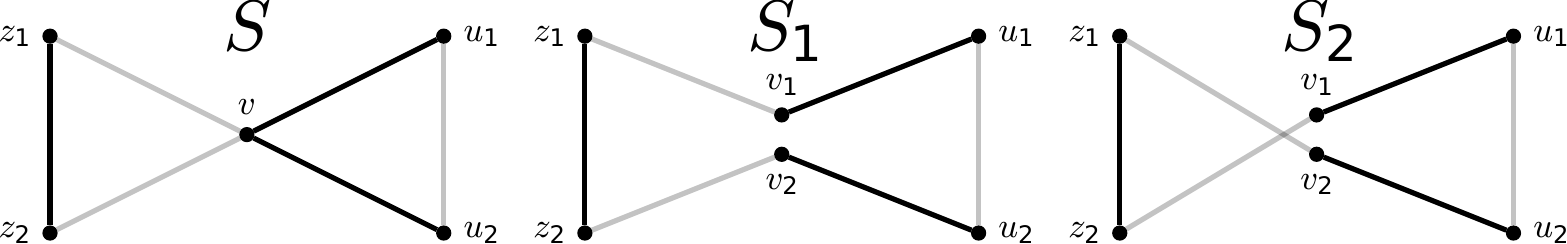} 
\caption{$S$ and its two Eulerian circles $S_{1}$ and $S_{2}$. We recover $S$ by merging $v_{1}$ and $v_{2}$.}
\label{simplecycles}
\end{figure} 

\section{A colored cost for 2-break scenarios}
\label{colorcost}

In this section we partition $G$'s vertices into subsets of different colors so as to define a \emph{colored cost function} $\varphi$ on 2-breaks: 
$\varphi((x_{1},x_{2}),(x_{3},x_{4})\rightarrow (x_{q_{1}},x_{q_{2}}),(x_{q_{3}},x_{q_{4}}))=0$ if 
\begin{align*}
\big\{\{col(x_{1}),col(x_{2})\},\{col(x_{3}),col(x_{4})\}\big\}=\big\{\{col(x_{q_{1}}),col(x_{q_{2}})\},\{col(x_{q_{3}}),col(x_{q_{4}})\}\big\},
\end{align*}
and $\varphi$ is 1 otherwise. 
For example a 2-break $(x,y),(z,t)\rightarrow (x,z),(y,t)$ is of zero cost if $col(x)=col(t)$ or $col(y)=col(z)$, otherwise it is of cost 1. 
This cost function is convenient for use with spacial proximity constraints
given by the packing of the chromosomes into the
nucleus~\cite{pulicani2017rearrangement}.
See Figure~\ref{0} for an example.

In what follows, by a \emph{graph} we will mean a graph together with a coloring $col$ of its vertices and by cost we will mean the cost $\varphi$ obtained using this coloring.
We define a \emph{color-merged graph} $J(G,col)$ obtained by merging the vertices of $G$ of the same color.
For a 2-break $G\rightarrow G'$ a transformation $J(G,col)\rightarrow J(G',col)$ is also a 2-break.
This means that a scenario $\rho$ for $G$ defines a scenario $\rho_{J}$ for $J(G,col)$.
If the cost of $G\rightarrow G'$ is $0$, 
then $J(G,col)=J(G',col)$, and such a move can be omitted from $\rho_{J}$, leaving us with a scenario of length equal to the cost of $\rho$.
This observation leads us to the following theorems proven in the Appendix. 

\vspace{-1mm}
\begin{theorem}
\label{mlsforgraph}
The minimum cost of a scenario for a graph $G$ is equal to $l(J(G,col))$. 
\end{theorem}

\begin{theorem}
\label{mcsnphard}
Deciding for a circle $O$ and a bound $k$ whether there exists a scenario of cost at most $k$ for $O$ is NP-hard.
\end{theorem}

\vspace{-6mm}
\subsection{Cycle decomposition of a color-merged graph}
\label{junctiondecomposition}

Edges of $G$ and $J(G,col)$ can be labeled in such a way that for the edges $(a,b)$ and $(x,y)$ labeled $i$ in $G$ and $J(G,col)$ we have equality $(x,y)=(col(a),col(b))$.
We say that such labelings \emph{conform}.
For a scenario $\rho$ of cost $w$ for $G$ we will construct a scenario $\rho_{J}$ for $J(G,col)$ of length $w$ that acts on the same edges as $\rho$.
We will mimic the process of Section~\ref{decompositionOFscenario},
taking more care this time when relabeling the edges during the scenario and mapping the black and gray edges at the end of the scenario.

We label the black edges of $G$ and $J(G,col)$ with $\{1,\ldots, e(G)\}$ and gray edges with $\{e(G)+1,\ldots, 2e(G)\}$.
Given a labeling $L_{G}$ we obtain a conforming labeling $L_{J}$ by merging the vertices of $G$ having the same color while keeping the labels of the edges. 
We set $P_{0}=\big\{\{1\},\ldots,\{e(G)\}\big\}, G_{0}=G$ and $J_{0}=J(G,col)$.
Let us fix a 2-break scenario $\rho$ of cost $w$ and length $m$ for $G$.
We take the $l$th 2-break $(a,b),(c,d)\rightarrow (a,c),(b,d)$ of $\rho$ transforming $G_{l-1}$ into $G_{l}$ with $L_{G_{l-1}}(a,b)=i$ and $L_{G_{l-1}}(c,d)=j$.

If a 2-break is of cost 1, then we label the newly added edges $(a,c)$ and $(b,d)$ with $i$ and $j$ respectively to obtain $L_{G_{l}}$ 
and merge the subsets in $P_{l-1}$ containing $i$ and $j$ to obtain $P_{l}$.
In $J(G_{l-1},col)$ we replace the edges labeled $i$ and $j$ by an edge $(col(a),col(c))$ labeled $i$ and $(col(b),col(d))$ labeled $j$ 
to obtain a labeling $L_{J_{l}}$ of $J(G_{l},col)$ that conforms to $L_{G_{l}}$.
If a 2-break is of cost 0, then without loss of generality we can suppose that $col(a)=col(d)$. 
In this case we label the newly added edges $(b,d)$ and $(a,c)$ with $i$ and $j$ respectively to obtain $L_{G_{l}}$. 
$J(G_{l-1},col)=J(G_{l},col)$ and $L_{G_{l}}$ is chosen in such a way that $L_{J_{l-1}}$ and $L_{G_{l}}$ still conform, thus we keep $L_{J_{l}}=L_{J_{l-1}}$ and $P_{l}=P_{l-1}$.

At the end of the scenario we obtain a partition $P_{m}$ of $\{1,\ldots, e(G)\}$.
Since $G_{m}$ is terminal, we can map the subsets $\{1,\ldots, e(G)\}$ and $\{e(G)+1,\ldots, 2e(G)\}$ one-to-one in a way that for a pair $(i,j)$ of mapped labels the black edge $i$ and gray edge $j$ have the same endpoints.
For $j\in\{e(G)+1,\ldots, 2e(G)\}$ we include it into a subset of $P_{m}$ containing a label to which $j$ is mapped.
This way a partition $R_{m}$ of a set $\{1,\ldots, 2e(G)\}$ is obtained.    
Only the 2-breaks of cost 1 of $\rho$ modify the colored-merged graph.
This provides us with a 2-break scenario $\rho_{J}$ of length $w$ for $J(G,col)$. 
In addition to that $R_{m}$, 
when seen as a partition of the edges of $J(G,col)$, is exactly the cycle decomposition $C(\rho_{J})$ of the scenario $\rho_{J}$.
In Section~\ref{circlecost} we study the structure of $C(\rho_{J})$ for a parsimonious scenario $\rho$.

\subsection{\mlps for a circle}
\label{circlecost}

For a circle $O$ we take conforming labelings $L_{O}$ and $L_{J}$ of $O$ and $J(O,col)$. 
For a subset $S$ of edges of $J(O,col)$ (resp. $O$) we define a subset $O(S)$ (resp. $J(S)$) of edges of $O$ (resp. $J(O,col)$) labeled with the same labels.
To a scenario $\rho$ for $O$ we have associated a scenario $\rho_{J}$ for $J(O,col)$ and its cycle decomposition $C(\rho_{J})$ in Section~\ref{junctiondecomposition}.

\begin{definition}[Crossing subsets]
Two disjoint subsets of edges $S_{1}$ and $S_{2}$ of a circle $O$ \emph{cross} 
if there are edges $e_{1},e_{2}$ in $S_{1}$ and $f_{1}, f_{2}$ in $S_{2}$ such that a path in $O$ joining $e_{1}$ and $e_{2}$ contains exactly one of the edges $f_{1}$ or $f_{2}$. 
A cycle decomposition $C$ of $J(O,col)$ is said to be non-crossing if none of the subsets of $O(C)$ cross. 
\end{definition}

\begin{theorem}
\label{parsimoniousTOdecomposition}
$C(\rho_{J})$ is non-crossing for a parsimonious scenario $\rho$ for $O$. 
\end{theorem}

\begin{proof}   
If the theorem is false, then the set of the circles for which there exists a parsimonious scenario contradicting the theorem is non-empty.
Let us take a circle $O^{*}$ in this set having the minimum number of edges and a scenario $\rho^{*}$ for $O^{*}$ such that $O(C(\rho^{*}_{J}))$ is crossing.   
$O^{*}$ has at least 2 black edges as otherwise $C(\rho^{*}_{J})$ contains a single subset. 

Using Theorem~\ref{minlength2break} and the structure of a circle we get that a 2-break of a parsimonious scenario for
a union of vertex-disjoint circles transforms one of its circles into a union of two vertex-disjoint circles.
This means that the first 2-break of $\rho^{*}$ replaces edges $i$ and $j$ and transforms $O^{*}$ into a union of two smaller circles $\bar{O}$ and $\hat{O}$. 
The following 2-breaks of $\rho^{*}$ replace 2 edges with labels belonging to either $\bar{O}$ or $\hat{O}$,
which provides us with the parsimonious scenarios $\bar{\rho}$ for $\bar{O}$ and $\hat{\rho}$ for $\hat{O}$.
By the minimality of $O^{*}$ we get that $C(\bar{\rho}_{\bar{J}})$ and $C(\hat{\rho}_{\hat{J}})$ are non-crossing. 
$C(\rho^{*}_{J})$ can be easily obtained from $C(\bar{\rho}_{\bar{J}})$ and $C(\hat{\rho}_{\hat{J}})$ by taking their union and then merging the subsets of edges including $i$ and $j$ if the first 2-break of $\rho^{*}$ is of cost 1. 
Now it is easy to check that $C(\rho^{*}_{J})$ is non-crossing, a contradiction.
\end{proof}
 
\begin{definition}[Arc] An arc is an alternating path joining two vertices of the same color, called \emph{endpoints}, and having the same number of black and gray edges.
Edges of an arc adjacent to its endpoints are called \emph{ends} and one of them is black and another is gray.
\end{definition}

A circle with $n$ vertices has $n$ arcs if all the vertices are of different colors and $n^2/2$ arcs if all the vertices share the same color.   
We say that two arcs do not \emph{overlap} if they are edge-disjoint or one is included in another but none of their ends coincide. 
A set of pairwise non-overlapping arcs will be called an \emph{independent} set of arcs. 
A \MISA (\misa) of $O$ is an independent subset of arcs of $O$ of maximum cardinality. 
For a set of arcs $S$ a \emph{maximal arc} is an arc that is not included in any other arc, 
while a \emph{minimal arc} is an arc that does not include any other arc. 
In the Appendix Lemma~\ref{decompositionTOmisa} is proven, leading us to Theorem~\ref{nestedscenario}.       

\begin{lemma} 
\label{decompositionTOmisa}
The size of a maximum non-crossing cycle decomposition of $J(O,col)$ is equal to the size of a \misa of $O$.   
\end{lemma}

\begin{theorem}
\label{nestedscenario}
The \mlps cost for a circle $O$ is equal to $e(O)-|I|$ for a \misa $I$ of $O$.
\end{theorem}

\begin{proof}
We first show that there exists a parsimonious scenario for $O$ of cost at most $e(O)-|I|$. 
A parsimonious scenario for $O$ is of length $e(O)-1$, thus of cost at most $e(O)-1$, due to Theorem~\ref{minlength2break}.
This means that if $|I|=1$, then inequality is trivial. 
Otherwise take a minimal arc $U$ in $I$.
Perform a 2-break replacing the black end of $U$ and the black edge not belonging to $U$ adjacent to the gray end of $U$ and transforming $O$ into a union of two circles.
One of these circles is of length equal to the length of $U$ and we call it $\hat{O}_{1}$.
Another is such that its \misa is of size $|I|-1$ and we call it $O_{1}$.
The cost of this 2-break is 0 as the endpoints of an arc have the same color. 
We iterate this procedure for $O_{1}$ until we end up with $|I|$ circles $O_{|I|-1}, \hat{O}_{1}, \ldots, \hat{O}_{|I|-1}$. 
A parsimonious scenario for this set of circles is of length $e(O)-|I|$ and thus of cost at most $e(O)-|I|$.

On the other hand, we take a \mlps scenario $\rho$ of cost $w$ for a circle $O$.
A cycle decomposition $C(\rho_{J})$ of $J(O,col)$ is of size at least $e(O)-w$ using Lemma~\ref{minlength2break}. 
Using Theorem~\ref{parsimoniousTOdecomposition} we obtain that $C(\rho_{J})$ is non-crossing and
using Lemma~\ref{decompositionTOmisa} we obtain an independent set of arcs of size $|C(\rho_{J})|$ 
and thus $w\geq e(O)-|C(J(\rho))|\geq e(O)-|I|$ where $I$ is a \misa.\end{proof}

\subsection{A polynomial time algorithm for \mlps on a circle}
\label{secPolyMCPS}

Fix a vertex $v$ of a circle $O$ having $n$ vertices. 
We say that an arc $U$ \emph{crosses} $v$ if $v$ is in $U$ without being its endpoint. 
We replace $v$ by two vertices of the same color transforming $O$ into an alternating path $P$.
\misa sizes of $P=(v_{0},\ldots,v_{l})$ and $O$ are equal due to Lemma~\ref{arcTOinterval}, proven in the Appendix.

\begin{lemma}
\label{arcTOinterval}
There exists a \misa $I'$ of $O$ such that there is no arc crossing $v$ in $I'$.
\end{lemma}

For $0\leq i\leq j\leq l$ $\misa(i,j)$ denotes the size of a \misa of an alternating path $(v_{i},\ldots, v_{j})$, $\misa(i,i)=0$ for all $i$.
We say that $i$ and $j$ are \emph{compatible} if $v_{i}$ and $v_{j}$ are the endpoints of an arc.
Take $I$ a \misa of a path $(v_{i},\ldots,v_{j})$ for $j\geq i$.  
Vertex $v_{j}$ belongs to at most one arc in $I$ and this creates three possibilities. 
If $v_{j}$ does not belong to any arc in $I$, then $|I|=\misa(i,j-1)$.
If an arc with the endpoints $v_{i}$ and $v_{j}$ is in $I$, then $|I|=\misa(i+1,j-1)+1$.
If an arc with the endpoints $v_{k}$ and $v_{j}$ is in $I$ with $i<k<j$, then $|I|=\misa(i,k)+\misa(k,j)$. This leads to the following recurrence:
\[\misa(i,j)=max 
\begin{cases}
    \misa(i,j-1),& \\
    \misa(i+1,j-1)+1,&\text{if }i\text{ and }j\text{ compatible}   \\
    \misa(i,k)+\misa(k,j),&\text{for }k,i<k<j\text{ compatible with }j\text{,}
\end{cases}\]
which provides us with a dynamic program with time complexity $O(n^3)$.
It is easy to modify the algorithm to give a \misa $I$ of $O$.
From $I$ it we can obtain a \mlps scenario for $O$ using Theorem~\ref{nestedscenario}.
Since $I$ is of size $O(n)$ this can easily be done in $O(n^2)$ time. 

We partition the vertices of $P$ into the subsets $S_{1},\ldots,S_{m}$ of pairwise compatible vertices and set $s_{i}=|S_{i}|$. 
One can show that our dynamic program computes $\misa(0,l)$ in $cn(s_{1}^{2}+\ldots+s_{m}^{2})$ steps for some constant $c$.
If the subsets are of equal sizes then the number of steps is $cn^{3}/m$.
This means that the best-case time-complexity of our program is $O(n^2)$.

\section{DCJ scenarios for genomes}
\label{mlpsbreakpoint}

A \emph{genome} consists of \emph{chromosomes} that are linear or circular orders of genes separated by potential \emph{breakpoint} regions.
In Figure~\ref{0} the tail of an arrow represents the \emph{tail extremity}, 
and the head of an arrow represents the \emph{head extremity} of a gene. 
We can represent a genome by a set of \emph{adjacencies} between the gene extremities. 
In Figure~\ref{0} this set is $\big\{\{1t\}, \{1h, 2t\}, \{2h,3h\}, \{3t\}\big\}$ for genome $A$ and $\big\{\{1t\}, \{1h, 2h\}, \{2t,3h\}, \{3t\}\big\}$ for genome $B$. 
An \emph{adjacency} is either an unordered pair of the extremities that are adjacent on a chromosome, 
called \emph{internal} adjacency, or a single extremity adjacent to one of the two ends of a linear chromosome, called an \emph{external} adjacency.
In what follows we will suppose two genomes $A$ and $B$ that share the same genes, and our goal will be to transform $A$ into $B$ using a sequence of DCJs.

\begin{definition}[Double cut and join]\label{dcjgenome} A DCJ cuts one or two breakpoint regions and joins the resulting ends of the chromosomes back in one of the four following ways: $\{a,b\},\{c,d\}\rightarrow\{a,c\},\{b,d\}$; $\{a,b\},\{c\}\rightarrow\{a,c\}$; $\{a,b\}\rightarrow\{a\},\{b\}$; $\{a\},\{b\}\rightarrow\{a,b\}$.
\end{definition}

We partition the gene extremities into subsets of different colors. $col(a)$ denotes the color of a gene extremity $a$.
The colored internal adjacency $\{a,b\}$ is $\{col(a), col(b)\}$ and the colored external adjacency $\{a\}$ is $\{col(a),\circ\}$, where $\circ$ does not coincide with any of the colors of the gene extremities.
A DCJ $A\rightarrow A'$ is said to be of zero cost if the sets of colored adjacencies of $A$ and $A'$ are equal. It is of cost 1 otherwise. 
For example $\{a,b\},\{c\}\rightarrow\{a,c\},\{b\}$ is of cost 0 if $\big\{\{col(a),col(b)\},\{col(c),\circ\}\big\}=\big\{\{col(a),col(c)\},\{col(b),\circ\}\big\}$, that is if $col(b)=col(c)$.
The cost of a DCJ scenario is the sum of the costs of its rearrangements. 

In~\cite{Bergeron2006}, a linear time algorithm for finding a parsimonious DCJ scenario was proposed.
The algorithm is based on the analysis of the connected components of the \emph{adjacency graph}.  
Here, we use a slightly different structure associated to a genome pair $(A,B)$ called the breakpoint graph~\cite{kececioglu1993exact,bafna1993genome,caprara1999sorting}.
 
\begin{definition}[Breakpoint graph]
$G(A,B)=(V, E^{b}\cup E^{g})$ for genomes $A$ and $B$, sharing $n$ genes, is a 2-edge-colored Eulerian undirected multi-graph. $V$ consists of $2n$ gene extremities and an additional vertex $\circ$. For every internal adjacency $\{a,b\}\in A$ (resp. $\{a,b\}\in B$) there is a black (resp. gray) edge $(a,b)$ in $G(A,B)$ and for every external adjacency $\{a\}\in A$ (resp. $\{a\}\in B$) there is a black (resp. gray) edge $(a,\circ)$ in $G(A,B)$. We add additional black and gray loops $(\circ, \circ)$ to obtain $d^{b}(\circ)=d^{g}(\circ)=2n$. The breakpoint graph $G(A,B)$ for the genomes from Figure~\ref{0} is given in Figure~\ref{1}.
\end{definition}

\begin{lemma}
\label{dcjTObrekapointgraph}
For the DCJ scenarios transforming genome $A$ into $B$ the minimum length is $l(G(A,B))=e(G(A,B))-c(G(A,B))$, the minimum cost is $l(J)=e(J)-c(J)$ with $J=J(G(A,B),col)$ and the minimum cost of a parsimonious DCJ scenario is $\mlps(G(A,B)$.
\end{lemma}

\begin{theorem} \mlps is polynomial-time solvable for a breakpoint graph $G(A,B)$.
\end{theorem}

\begin{proof}
Take genomes $A$ and $B$ sharing $n$ genes.
For all the vertices $v\neq\circ$ we have $d^{g}(G(A,B),v)=d^{b}(G(A,B),v)=1$.
From this we obtain that for an edge belonging to a circle this is the only simple cycle in $G(A,B)$ including this edge. 
Thus a \macd of $G(A,B)$ includes all of its circles.
These set aside, we are left with $G(A,B)'$, which is a union of alternating paths starting and ending at $\circ$ and having the end edges of the same color. 
If this color is black we call a path $AA$, and $BB$ otherwise.
Every simple cycle of $G(A,B)'$ is a union of a $BB$ path and a $AA$ path.
We proceed by constructing a complete bipartite graph $H$ having $AA$ and $BB$ paths as vertices. 
An edge joining paths $a$ and $b$ is assigned the weight equal to the \mlps cost of $a\cup b$.  
\begin{lemma}
\label{weightinggraph} 
The weights of the edges of $H$ can be assigned in $O(n^4)$ time. 
\end{lemma}
Lemma~\ref{weightinggraph} is proven in the Appendix. 
We proceed by computing a maximum weight matching for $H$ to obtain a \macd of $G(A,B)'$ minimizing \mlps cost, this can be done in $O(n^3)$ time using Hungarian algorithm.   
The \mlps cost of $G(A,B)$ is obtained by adding the costs of the circles and the cost of $G(A,B)'$.
In this way we obtain a $O(n^4)$ time algorithm for computing the \mlps cost, which it can be easily modified to give a \mlps scenario. 
\end{proof}
\begin{figure}[!h]
\centering
\includegraphics[scale = 0.6]{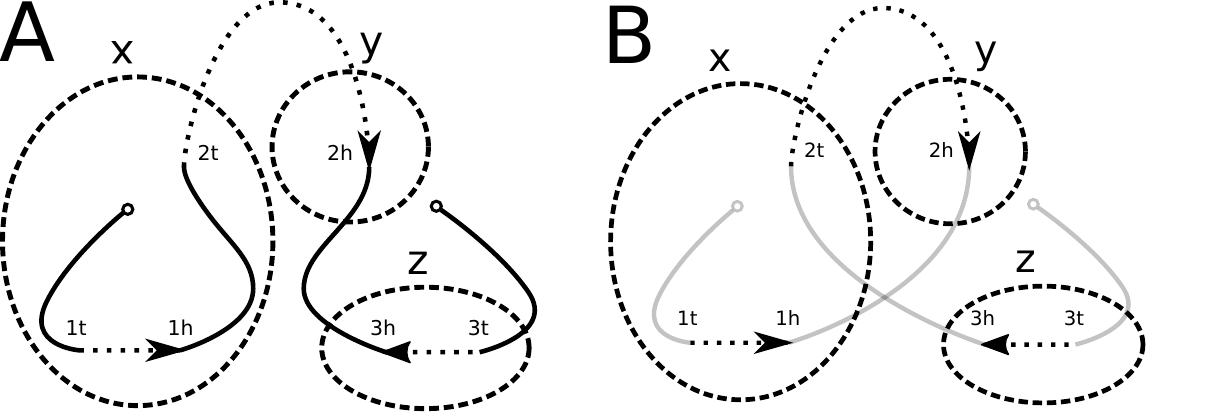} 
\caption{Genomes $A$ and $B$ share genes $1,2$ and $3$ and both consist of a single linear chromosome. Their gene extremities are colored in three colors $x,y$ and $z$ according to spacial proximity.
The DCJ $\{1_{h},2_{t}\},\{2_{h},3_{h}\}\rightarrow\{1_{h},2_{h}\},\{2_{t},3_{h}\}$ transforming $A$ into $B$ is of cost 1 because $\big\{\{x,x\},\{y,z\}\big\}\neq\big\{\{x,y\},\{x,z\}\big\}$. 
The DCJ $\{2_{t},1_{h}\},\{1_{t}\}\rightarrow \{2_{t},1_{t}\},\{1_{h}\}$ on genome $A$ or $B$ would be of cost 0.}
\label{0}
\end{figure}
\begin{figure}[!h]
\centering
\includegraphics[scale = 0.8]{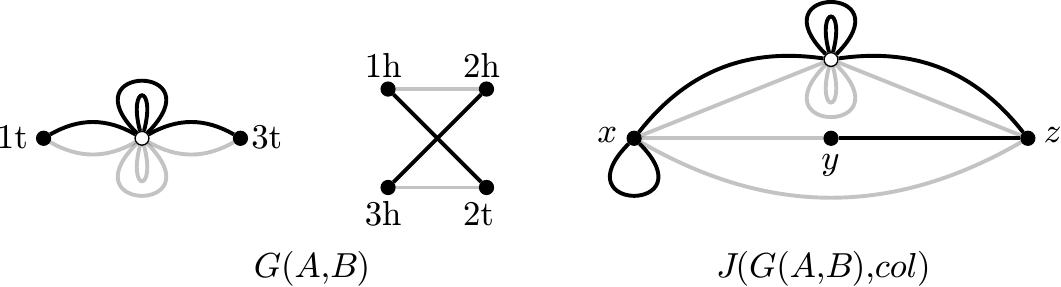} 
\caption{Breakpoint graph $G(A,B)$ and its color-merged graph $J=J(G(A,B),col)$ for the genomes and a coloring of their gene extremities given in Figure~\ref{0}.
$e(J)=6$ and $c(J)=5$.}
\label{1}
\end{figure}

\section{Conclusions and further work}

\vspace{-1mm}
\subsection{Sorting by mathematical transpositions}

Finding a parsimonious 2-break scenario on a circle is closely related to the problem of sorting a circular permutation by mathematical transpositions. 
For a permutation $\pi$ of $\{1,\ldots, n\}$ we define a digraph $D(\pi)$ with vertices $i$ and directed edges $(i,\pi(i))$. 
A transposition on $\pi$ defines a 2-break on $D(\pi)$ as illustrated in Figure~\ref{-2}. 
The problem of finding a minimum cost parsimonious scenario of transpositions for permutations with a cost function $\varphi$ defined for every pair of elements of $V$ was treated in~\cite{Farnoud(Hassanzadeh):2012:SPC:2331863.2332291}. 
Such a $\varphi$ defines a natural cost function $\varphi'$ for the 2-breaks on $V$ such that $\varphi'(\{a,b\},\{c,d\}\rightarrow\{a,c\},\{b,d\})=min(\varphi(a,d),\varphi(b,c))$.
The cost function used in Section~\ref{colorcost} is precisely of such a type. 
The paper described a polynomial-time algorithm for a minimum cost parsimonious scenario of transpositions for a general cost function $\varphi$.
We speculate that this algorithm can be adapted to obtain a polynomial algorithm for $\mlps_{\varphi'}$ on a circle.
Another problem of interest is finding a minimum cost scenario among the scenarios of length smaller than some fixed length. 
This would be an important step towards finding more realistic evolutionary scenarios, 
since the most likely scenario may not always be of minimum length.

\begin{figure}[h]
\centering
\includegraphics[scale = 0.7]{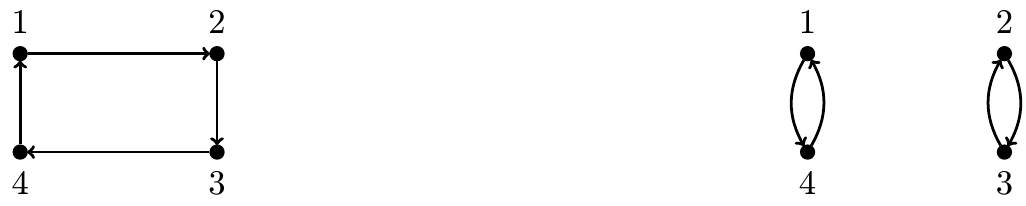} 
\caption{For a permutation $\pi=2341$ digraphs $D(\pi)$ and $D((24)\pi)$ are given. A transformation $D(\pi)\rightarrow D((24)\pi)$ is a 2-break $(1,2),(3,4)\rightarrow(1,4),(3,2)$.}
\label{-2}
\end{figure}

\vspace{-4mm}
\subsection{Conclusions}

The DCJ models on genes that are oriented, or unoriented~\cite{Chen2013}, have
insertions and deletions~\cite{Shao2012}, or have segmental
duplications~\cite{shao2015comparing}, are all intimately tied to the
breakpoint graph. They all can be easily formulated in our setting of 2-edge-colored Eulerian
multi-graphs and 2-break scenarios; all of our results about cycle decompositions and
\mlps apply directly in these cases.
We showed one example of how to use our work algorithmically, but we expect 
our framework to lead to further algorithmic results on general cost functions
and general DCJ distances that consider unequal content.

\bibliography{cpm2018}

\newpage

\appendix

\subsection{Lemma~\ref{degrees}}

\begin{lemmawithoutnumber}
For a vertex $v$ and a subset $Q\in R_{l}$ with $l\in\{0,\dots m\}$ we have
\begin{align*}
\bar{d}(s(G_{0},Q),v)=\bar{d}(s(G_{l},Q),v)
\end{align*}
\end{lemmawithoutnumber}   

\begin{proof}
Equality is true for $l=0$. Suppose that equality is true for every $Q$ and $v$ with $l-1$ and proceed by induction on $l>0$. Fix a vertex $v$ and a subset $Q\in R_{l}$. The $l$th k-break replaces the edges labeled $e_{1},\ldots,e_{k}$. 
By construction, these edges belong to the same subset $Q'\in R_{l}$. There are two possibilities:
\begin{itemize}
\item ($Q'\neq Q$) In this case $Q\in R_{l-1}$ and $s(G_{l},Q)=s(G_{l-1}, Q)$, as the edges in $Q$ are unaffected by the $l$th k-break. Using the inductive hypothesis we obtain
\begin{equation*}
  \bar{d}(s(G_{l},Q),v)=\bar{d}(s(G_{l-1},Q),v)=\bar{d}(s(G_{0},Q),v)
\end{equation*}
\item ($Q=Q'$) In this case $Q'=\bigcup_{i=1}^{k} Q_{i}$ where $e_{i}\in Q_{i}\in R_{l-1}$. 
There may exist $i,j \in\{1,\ldots,k\}$ such that $Q_{i}=Q_{j}$,
thus we select such $q_{1},\ldots,q_{r}\in\{1,\ldots,k\}$ so that $Q=Q'=\bigcup_{i=1}^{r} Q_{q_{i}}$ and all of these subsets are different. To begin with,
\begin{align*}
  \bar{d}(s(G_{l},Q),v) = \bar{d}(s(G_{l-1},Q),v) = \sum_{i=1}^{r}\bar{d}(s(G_{l-1},Q_{q_{i}}),v).
\end{align*}
The first equality is due to the fact that a graph $s(G_{l},Q)$ is obtained from $s(G_{l-1},Q)$ by a k-break and k-break does not modify the degrees of the vertices in a graph. The second equality is guaranteed since no two $Q_{q_{1}},\ldots, Q_{q_{r}}$ intersect. Then 
\begin{align*}
 \sum_{i=1}^{r}\bar{d}(s(G_{l-1},Q_{q_{i}}),v) = \sum_{i=1}^{r}\bar{d}(s(G_{0},Q_{q_{i}}),v)=\bar{d}(s(G_{0},Q),v).
\end{align*}
With the first equality following from the inductive hypothesis and the latter once again since no two $Q_{q_{1}},\ldots, Q_{q_{r}}$ intersect.
\end{itemize}
As equality is preserved by a k-break, and true for $l=0$, we obtain the result by induction.
\end{proof}

\subsection{Theorem~\ref{costOFsimple}}

\begin{theoremwithoutnumber} 
The $\mlps_{\varphi}$ cost of a simple cycle $S$ is equal to the minimum over all $\mlps_{\varphi'}$ costs of its Eulerian circles. 
\end{theoremwithoutnumber}

\begin{proof}

Take a simple cycle $S$ with $d(S)=k$. 
If $k=0$, then $S$ is its own Eulerian circle. 
If $k>0$, then we choose a vertex $v$ satisfying $d^{b}(S,v)=d^{g}(S,v)=2$ and construct a graph with $v$ replaced by two vertices $v_{1}$ and $v_{2}$ satisfying $d^{b}(S,v_{i})=d^{g}(S,v_{i})=1$.
We start with a copy of $S$ from which we remove $v$ with the adjacent edges and add two new vertices $v_{1}$ and $v_{2}$ to obtain $S'$.
If there is a black loop $(v,v)$ in $S$ then we add a black edge $(v_{1},v_{2})$ to $S'$.  
Otherwise there are two black edges $(v,u_{1})$ and $(v,u_{2})$ in $S$ with possibly $u_{1}=u_{2}$. 
We add black edges $(v_{1},u_{1})$ and $(v_{2},u_{2})$ to $S'$.
If there is a gray loop $(v,v)$ in $S$, then we add a gray edge $(v_{1},v_{2})$ to obtain a graph $S_{0}$.
If there are two gray edges $(v,z_{1})$ and $(v,z_{2})$ in $S$, then we construct two graphs $S_{1}$ and $S_{2}$.
$S_{1}$ is obtained by adding gray edges $(z_{1},v_{1})$ and $(z_{2},v_{2})$ to $S'$ and $S_{2}$ by adding gray edges $(z_{1},v_{2})$ and $(z_{2},v_{1})$ to $S'$.
In Figure~\ref{simplecycles} the graphs $S_{1}$ and $S_{2}$ are given for a simple cycle $S$.

$S_{0}, S_{1}, S_{2}$ have the same set of vertices $V'$. 
We set $\bar{v}_{1}=\bar{v}_{2}=v$ and $\bar{u}=u$ for other vertices in $V'$.     
We define a cost function $\varphi'$ for the 2-breaks on $V'$ as follows:
\begin{align*}
\varphi'((x_{1},x_{2}),(x_{3},x_{4})\rightarrow(x_{1},x_{3}),(x_{2},x_{4}))=\varphi((\bar{x}_{1},\bar{x}_{2}),(\bar{x}_{3},\bar{x}_{4})\rightarrow(\bar{x}_{1},\bar{x}_{3}),(\bar{x}_{2},\bar{x}_{4})).
\end{align*}

From $S_{i}$ we obtain $S$ by merging the vertices $v_{1}$ and $v_{2}$. 
We will show that $\mlps_{\varphi}(S)$ of $S$ is equal to the $\mlps_{\varphi'}(S_{0})$ or the minimum 
of $\mlps_{\varphi'}(S_{1})$ and $\mlps_{\varphi'}(S_{2})$ depending on whether $S$ contains a gray loop $(v,v)$ or not. 
$c(S_{i})=1$ and $e(S_{i})=e(S)$, thus the lengths of their parsimonious scenarios are equal due to Theorem~\ref{minlength2break} .    

For every 2-break transforming $S_{i}$ to $S_{i}'$ there is a unique 2-break on $S$ of the same cost transforming $S$ into a graph that we obtain from $S_{i}'$ by merging $v_{1}$ and $v_{2}$. 
If we take a \mlps scenario for $S_{i}$ we obtain a parsimonious scenario for $S$ of the same cost.
This means that $\mlps_{\varphi}(S)$ is smaller or equal to $\mlps_{\varphi'}(S_{0})$ or the minimum 
of $\mlps_{\varphi'}(S_{1})$ and $\mlps_{\varphi'}(S_{2})$ depending on whether $S$ has a gray loop $(v,v)$ or not. 

On the other hand for a 2-break transforming $S$ into $S'$ there exists a 2-break of the same cost transforming $S_{i}$ into such $S_{i}'$ that by merging its vertices $v_{1}$ and $v_{2}$ we obtain $S'$. 
A \mlps scenario for $S$ provides us with a sequence of 2-breaks $\rho$ on $S_{i}$ of the same cost transforming it into a graph $\hat{S_{i}}$ from which we obtain a terminal graph by merging $v_{1}$ and $v_{2}$.

The structure of $\hat{S_{i}}$ is fairly simple and the two possible cases can be checked by hand. 
If $S$ contains a gray loop $(v,v)$ then $\hat{S_{0}}$ must be already terminal as it is Eulerian and there is a single gray edge adjacent to $v_{0}$ or $v_{1}$. 
In this case a \mlps scenario for $S$ and $\varphi$ provides us with a parsimonious scenario of the same cost for $S_{0}$ and $\varphi'$. 

Sets of black edges of $S_{1}$ and $S_{2}$ are equal.
This means that a \mlps scenario for $S$ provides us with a single sequence $\rho$ of 2-breaks of the same cost transforming $S_{1}$ into $\hat{S_{1}}$ and $S_{2}$ into $\hat{S_{2}}$ such that by merging $v_{1}$ and $v_{2}$ we obtain terminal graphs.
The sets of black edges of $\hat{S_{1}}$ and $\hat{S_{2}}$ stay equal. 
If $\hat{S_{1}}$ is already terminal, then we are done. Otherwise $\hat{S_{1}}$ is a union of a terminal graph and a cycle of length 2 containing gray edges $(z_{1},v_{1})$ and $(z_{2},v_{2})$ and black edges $(z_{2}, v_{1})$ and $(z_{1}, v_{2})$, however this means that $\hat{S_{2}}$ is terminal. Thus in this case $\mlps_{\varphi}(S)$ is equal to the minimum of $\mlps_{\varphi'}(S_{1})$ and $\mlps_{\varphi'}(S_{2})$.    

Either $S_{i}$ are circles and we are done or we can proceed by choosing another vertex $v$.
At the end we obtain Eulerian circles of $S$ and $\mlps_{\varphi}(S)$ is equal to the minimum of the $\mlps$ costs for these circles.  
\end{proof}

\subsection{Theorem~\ref{mlsforgraph}}
We start by proving an auxiliary lemma. 

\begin{lemmawithoutnumber}
If $J(G,col)$ is terminal, then there exists a zero cost 2-break scenario for $G$.  
\end{lemmawithoutnumber}

\begin{proof}
We denote the maximum number of length 1 cycles in a cycle decomposition of $G$ by $c_{1}(G)$. 
If $c_{1}(G)=e(G)$, then $G$ is terminal and we are done. 
Otherwise we demonstrate how we can transform $G$ with a sequence of zero cost 2-breaks into $G'$ with $c_{1}(G')>c_{1}(G)$.
When iterated this provides us with a zero cost scenario for $G$. 

Take a cycle decomposition $C$ of $G$ having $c_{1}(G)$ length 1 cycles and remove these $c_{1}(G)$ cycles from $G$ obtaining $\bar{G}$. 
Take a black edge $(u,v)$ of $\bar{G}$.
$J(\bar{G},col)$ is terminal, thus there exists a gray edge $(col(u),col(v))$ in $J(\bar{G},col)$.
This means that there is a gray edge $(u',v')$ in $\bar{G}$ with $col(u)=col(u')$ and $col(v)=col(v')$. 
$G$ is Eulerian thus $\bar{G}$ is also Eulerian and there are gray edges $(p,u')$ and $(r, v')$ in $\bar{G}$. 
\begin{itemize}
\item If $(u,v)$ equals one of these edges, let's say $(p,u')$, then $col(p)=col(v)=col(v')$.
This means that a 2-break $(p,u'),(r,v')\rightarrow (p,r),(u',v')$ transforming $\bar{G}$ into $\bar{G}'$ is of zero cost and creates a length 1 cycle. 
\item If $(u,v)$ does not coincide with $(p,u')$ or $(r, v')$, then a 2-break $(u,v),(p,u')\rightarrow (u,p),(v,u')$ is of zero cost as $col(u)=col(u')$.
So is a 2-break $(v,u'),(v',r)\rightarrow (v,r),(u',v')$ as $col(v)=col(v')$. 
\end{itemize}
These 2-breaks transforming $\bar{G}$ into $\bar{G}'$ are of zero cost and create a length 1 cycle in $\bar{G}'$. 
Once the length 1 cycles deleted from $G$ are reintroduced to $\bar{G}'$ we obtain a graph $G'$ with $c_{1}(G')>c_{1}(G)$ and a sequence of zero cost 2-breaks transforming $G$ into $G'$.
\end{proof}

\begin{theoremwithoutnumber}
The minimum cost of a scenario for a graph $G$ is equal to $l(J(G,col))$. 
\end{theoremwithoutnumber}

\begin{proof}
For a 2-break $G\rightarrow G'$ a transformation $J(G,col)\rightarrow J(G',col)$ is also a 2-break.
If the cost of $G\rightarrow G'$ is $0$, then $J(G,col)=J(G',col)$. 
This means that for a scenario of cost $w$ for $G$ there exists a scenario of length $w$ for $J(G,col)$ and thus $l(J(G,col))\leq w$. 
On the other hand, for every 2-break $J(G,col)\rightarrow J'$ a 2-break $G\rightarrow G'$ can be found such that $J(G',col)=J'$. 
For $J(G,col)$ scenario of length $l=l(J(G,col))$ we obtain a 2-break sequence of length $l$, thus of cost at most $l$, 
transforming $G$ into such $G'$ that $J(G',col)$ is terminal. 
Using the previously proven lemma we get a scenario for $G$ of cost at most $l$ establishing $l(J(G,col))\geq w$.  
\end{proof}

\subsection{Theorem~\ref{mcsnphard}}
\begin{theoremwithoutnumber}
Deciding for a circle $O$ and a bound $k$ whether there exists a scenario of cost at most $k$ for $O$ is NP-hard.
\end{theoremwithoutnumber}

\begin{proof}
The problem is clearly in NP.
We reduce the decision version of a maximum cycle decomposition on simple Eulerian graphs, which is
NP-hard~\cite{Holyer81}, to our problem.
Without loss of generality, take an instance $G=(V,E)$ and a bound $k$,
where $G$ is Eulerian and connected.
Consider an Eulerian cycle $(u_{1},u_{2},\dots,u_{n},u_{1})$ of $G$ and construct
a circle $O=(v_{1},\ldots,v_{2n},v_{1})$ with gray edges $(v_{2i-1},v_{2i})$, black edges $(v_{2i},v_{2i+1})$ and 
$col(v_{2i-1})=col(v_{2i})=u_{i}$ for $i\in\{1,\ldots, n\}$.
The gray edges of $J(O,col)$ are loops and black edges of $J(O,col)$ are exactly the edges of $G$.
There is a scenario of cost at most $e(G)-k$ for $O$ if and only if there exists an alternating cycle packing of $J(O,col)$ of size at least $k$ due to Theorem~\ref{mlsforgraph} and Theorem~\ref{minlength2break}.
And the latter is true if and only if $G$ admits a maximum cycle decomposition of size least $k$.
\end{proof}

\subsection{Lemma~\ref{decompositionTOmisa}}

We start by proving an auxiliary lemma. 

\begin{lemmawithoutnumber}
\label{every}
Every edge of a circle $O$ is included in at least one arc of its \misa $I$.    
\end{lemmawithoutnumber}
\begin{proof}
Let us direct the edges of $O$ to obtain a directed cycle and suppose by contradiction that an edge $e$ entering a vertex of color $x$ does not belong to any arc in $I$. 
Without loss of generality we can suppose that the color of $e$ is black. 
The number of gray edges leaving the vertices of color $x$ is equal to the number of black edges entering the vertices of color $x$ in $O$ and in any arc of $O$.
This is true due to the constraint that an arc must have the vertices of the same color as its endpoints and the edges of different colors as its ends. 
This means that there is a gray edge $f$ in $O$ not included in any arcs of $I$ leaving a vertex of color $x$.
We define a new arc starting and ending at the vertices of color $x$ with edges $e$ and $f$ as its ends. 
By construction, this new arc does not overlap with any of the arcs in $I$, which contradicts the maximality of $I$.  
\end{proof}

Let us take a set $D$ of edge-disjoint non-crossing cycles of $J(O,col)$ and an independent set of arcs $Y$ of $O$. 
We call $(D,Y)$ a valid pair if the set of edges included in $Y$ and $O(D)$ partition the set of edges of $O$. 

\begin{lemmawithoutnumber} 
The size of a maximum non-crossing cycle decomposition of $J(O,col)$ is equal to the size of a \misa of $O$.   
\end{lemmawithoutnumber}

\begin{proof}
For a \misa $I$ of $O$ we have that $(\emptyset,I)$ is a valid pair due to the previously proven lemma. 
We start by showing how from a given valid pair $(D,Y)$ with a non-empty $Y$ we can obtain a valid pair $(D',Y')$ with $|D'|=|D|+1$ and $|Y'|=|Y|-1$. 
Take a maximal arc $U$ in $Y$. 
The set $c$ of edges of $U$ not belonging to any other arc in $Y$ is non-empty as it includes its ends. 
$J(V)$ is a cycle of $J(O,col)$ for every arc $V$ of $O$, 
thus we obtain that $J(c)$ is also a cycle of $J(O,col)$ and, by construction, it does not cross with the cycles in $D$. 
We include this cycle in $D$ and eliminate $U$ from $Y$ to obtain $(D',Y')$. 
Iterating this step we obtain a non-crossing cycle decomposition of $J(O,col)$ of size $|I|$.

Take a maximum non-crossing cycle decomposition $C$ of $J(O,col)$. $(C, \emptyset)$ is a valid pair.
We proceed by showing how from a given valid pair $(D,Y)$ with non-empty $D$ we can obtain a valid pair $(D',Y')$ with $|D'|=|D|-1$ and $|Y'|=|Y|+1$.   
For every cycle $c$ in $D$ we define the minimum length path $P(c)$ in $O$ that contains every edge of $O(c)$. 
We choose such $c$ for which $P(c)$ is of minimum length.
Since $P(c)$ is the minimum length path containing $O(c)$, we have that the ends of $P(c)$ belong to $O(c)$ and thus do not belong to any arc in $Y$. 
We remove from $P(c)$ any arcs of $Y$ that it might contain to obtain $P'$, a set of edges of $O$, 
and show that $P'$ consists entirely of edges of $O(c)$. 
Suppose by contradiction that there is an edge $e\in P'$ such that $e\notin O(c)$.
As it does not belong to any arc in $Y$ and $(D,Y)$ is valid, there must be a cycle $c'\in C$ that contains $e$.
However $O(c)$ and $O(c')$ can not cross and the endpoints of $P(c)$ belong to $O(c)$ and stay in $P'$, thus $O(c')$ must be properly included
in $P'$ and thus in $P(c)$. 
However this means that $P(c')$ is also properly included in $P(c)$ and thus shorter in length which contradicts the minimality of $P(c)$. 
Thus we obtain that $P'$ consists entirely of edges of $O(c)$. 
Now we show that $P(c)$ is an arc non-overlapping with any arc in $Y$. 
We have already shown that $P(c)$ is a union of $O(c)$ where $c$ is a cycle in $J(O,col)$ and maybe some arcs from $Y$.
This establishes that there is an equal number of gray and black edges in $P(c)$ and that the colors of its endpoints are the same, 
which means that $P(c)$ is an arc. 
By construction, the ends of this arc do not belong to any arc in $Y$ and thus $\{P(c)\}\cup Y$ is an independent set of arcs. 
We remove $c$ from $C$ and add $P(c)$ to $Y$ to obtain a valid pair $(D',Y')$. 
Iterating this step we obtain an independent set of arcs of $O$ of cardinality equal to the size of a maximum non-crossing cycle decomposition.\end{proof}

\subsection{Lemma~\ref{arcTOinterval}}

\begin{definition}[Chain]
A chain $H$ in the set of arcs $I$ is a set of edge-disjoint arcs $U_{1},\ldots, U_{m}$ in $I$ with endpoints of $U_{i}$ being $u_{i}$ and $v_{i}$ such that all these endpoints are of the
same color and $v_{i}=u_{i+1}$ for $i\in\{1,\ldots,m-1\}$. We say that $u_{0}$ and $v_{m}$ are the endpoints of $H$.  
\end{definition}

\begin{lemmawithoutnumber}
There exists a \misa $I'$ of $O$ such that there is no arc crossing $v$ in $I'$.
\end{lemmawithoutnumber}

\begin{proof}
If there is no \misa having an arc crossing $v$ then we are done. 
Otherwise take a \misa $I$ of $O$ having at least one arc crossing $v$.
All such arcs include the edges adjacent to $v$, thus there exists the single maximal arc $U$ in $I$ crossing $v$ that we remove from $I$.
If after this $v$ does not belong to any arc in $I$ then we can include in $I$ an arc whose both endpoints are $v$ to obtain an independent set of arcs $I'$ with no arc crossing $v$. 
Otherwise we take a maximum length chain $H$ in $I$ including $v$. 
Due to the maximality of $U$ we obtain that $H$ is included in $U$ and does not include its ends, 
which means that the endpoints of $H$ do not coincide. 
We take an edge $e$ of $O$ adjacent to an endpoint of $H$ but not belonging to $H$ and show that $e$ does not belong to any arc in $I$.
$e$ can not be an end of an arc as this arc would intersect with an arc in $H$ or could extend it. 
If $e$ belongs to an arc $V$ without being its end, then all of $H$ must belong to $V$, which contradicts the maximality of $H$.
We obtain that an arc joining the endpoints of $H$ and not crossing $v$ can be included in $I$ giving an independent set of arcs $I'$ equal in size to $I$ and having the number of arcs crossing $v$ smaller than $I$. Iterating this process we obtain a \misa with no arcs crossing $v$.     
\end{proof}

\subsection{Lemma ~\ref{dcjTObrekapointgraph}}
\begin{lemmawithoutnumber}
For the DCJ scenarios transforming genome $A$ into $B$ the minimum length is $l(G(A,B))=e(G(A,B))-c(G(A,B))$, the minimum cost is $l(J)=e(J)-c(J)$ with $J=J(G(A,B),col)$ and the minimum cost of a DCJ scenario of minimum length is $\mlps(G(A,B)$.
\end{lemmawithoutnumber}
\begin{proof}
$G(A,B)$ is constructed in such a way that for every DCJ $A\rightarrow A'$ the transformation $G(A,B)\rightarrow G(A',B)$ is a 2-break. Notably, a DCJ $\{a,b\}\rightarrow\{a\},\{b\}$ results in a transformation $(a,b),(\circ,\circ)\rightarrow (a,\circ),(b,\circ)$, as the construction of a breakpoint graph guarantees that there are enough black loops $(\circ,\circ)$ to realize such a 2-break. For any 2-break $G(A,B)\rightarrow G'$ with $G'\neq G(A,B)$ there exists a DCJ $A\rightarrow A'$ such that $G(A',B)=G'$. Since $G(B,B)$ is terminal, it follows that the minimum length of a scenario transforming $A$ into $B$ is $l(G(A,B))$.

Vertices of $G(A,B)$ are the gene extremities of $A$ and $B$ plus an additional vertex $\circ$. 
We color $\circ$ with a unique color.
The rest of the vertices of $G(A,B)$ have the colors of their gene extremities.
By construction, the cost of a DCJ $A\rightarrow A'$ is equal to the cost of a 2-break $G(A,B)\rightarrow G(A',B)$.
On the other hand for a 2-break $G(A,B)\rightarrow G'$ there exists a DCJ $A\rightarrow A'$ of the same cost such that $G(A',B)=G'$.
This means that the minimum cost of a scenario transforming $A$ into $B$ is equal to the minimum cost of a scenario for $G(A,B)$ and we conclude using Theorem~\ref{mlsforgraph} that this cost is $l(J)=e(J)-c(J)$ with $J=J(G(A,B),col)$. 

Let us denote the minimum cost of a parsimonious DCJ scenario transforming $A$ into $B$ by $w_{min}$.
In the two previous paragraphs we have seen that for a DCJ scenario of length $l$ and cost $w$ transforming $A$ into $B$ there exists a 2-break scenario on  
$G(A,B)$ of length $l$ and cost $w$. This means that $\mlps(G(A,B))\leq w_{min}$. 
On the other hand we have also seen that for a 2-break scenario on $G(A,B)$ of length $l$ and cost $w$ there exists a DCJ scenario of length $l$ and cost $w$ transforming $A$ into $B$ and this establishes the equality $\mlps(G(A,B))=w_{min}$.
\end{proof}

\subsection{Lemma~\ref{weightinggraph}}
\begin{lemmawithoutnumber}
The weights of the edges of $H$ can be assigned in $O(n^4)$ time. 
\end{lemmawithoutnumber}
\begin{proof}$a_{1},\ldots,a_{x}$ denotes the sizes of $AA$ paths and $b_{1},\ldots,b_{y}$ denotes the sizes of $BB$ paths with
$\sum_{i=0}^{x}a_{i}=|AA|$ and $\sum_{j=0}^{y}b_{j}=|BB|$.
By construction, $d^{b}(G(A,B),\circ)=d^{g}(G(A,B),\circ)=2n$, meaning that $x,y\leq n$. 
\mlps cost of a union of two paths having $a_{i}$ and $b_{j}$ vertices can be computed in $2c(a_{i}+b_{j})^{3}$ steps as \mlps cost of a circle of size $n$ can be computed in $cn^3$ steps for some constant $c$ and we need to compute this cost for two Eulerian circles. 
We can compute \mlps cost for every pair of $AA$ and $BB$ path in the number of steps equal to
\begin{align*}
\sum_{i=0}^{x}\sum_{j=0}^{y}2c(a_{i}+b_{j})^{3}=&2c\sum_{i=0}^{x}\sum_{j=0}^{y} (a_{i}^3+3a_{i}^{2}b_{j}+3b_{j}^{2}a_{i}+b_{j}^3)\\
=&2c(y\sum_{i=0}^{x}a_{i}^{3}+x\sum_{j=0}^{y}b_{j}^{3}+3|BB|\sum_{i=0}^{x}a_{i}^{2}+3|AA|\sum_{j=0}^{x}b_{j}^{2})
\end{align*}
The terms $y,x,|AA|$ and $|BB|$ are clearly $O(n)$ and we obtain the worst-case time-complexity $O(n^4)$ for weighting the bipartite graph.
\end{proof}

\end{document}